\documentclass[11pt]{article}
\usepackage{xspace}
\usepackage{graphicx, epsfig}
\usepackage{amsmath,amssymb}

\usepackage{fullpage}

 \newtheorem{lemma}{Lemma}
  \newtheorem{corollary}{Corollary}
 
  \newtheorem{theorem}{Theorem}

\newcommand{\halfrightsect}[2]{[#1,#2)}
\newcommand{\halfleftsect}[2]{(#1,#2]}

\newcommand{\eps}{\varepsilon}

\newcommand{\tO}{\widetilde{O}}

\newcommand{\pred}{\mathrm{pred}}
\newcommand{\successor}{\mathrm{succ}}
\newcommand{\rank}{\mathrm{rank}}

\newenvironment{proof}{\trivlist\item[]\emph{Proof}:}%
{\unskip\nobreak\hskip 1em plus 1fil\nobreak$\Box$
\parfillskip=0pt%
\endtrivlist}

{\unskip\nobreak\hskip 1em plus 1fil\nobreak$\Box$
\parfillskip=0pt%
\endtrivlist}

\begin{document}
\title{Space Efficient Multi-Dimensional Range Reporting}
\author{Marek Karpinski\thanks{Email {\tt marek@cs.uni-bonn.de}.} \and 
Yakov Nekrich\thanks{Email {\tt yasha@cs.uni-bonn.de}.} 
}
\date{Dept. of Computer Science\\ University of Bonn.}
\maketitle

\begin{abstract}
\tolerance=1000
We present a data structure that supports three-dimensional 
range reporting queries in $O(\log \log U + (\log \log n)^3+k)$ time and uses 
$O(n\log^{1+\eps} n)$ space, where $U$ is the size of the universe, 
$k$ is the number of points in the answer,
 and $\eps$ is an arbitrary constant. This result improves over  the data 
structure 
of Alstrup, Brodal, and Rauhe (FOCS 2000) that uses $O(n\log^{1+\eps} n)$ 
space 
and supports queries in $O(\log n+k)$ time, the data structure of Nekrich 
(SoCG'07) that uses $O(n\log^{3} n)$ space and supports queries in 
$O(\log \log U + (\log \log n)^2 + k)$ time,   and the data structure of 
Afshani (ESA'08) that uses $O(n\log^{3} n)$ space and also supports queries 
in 
$O(\log \log U + (\log \log n)^2 + k)$ time but relies on randomization 
during the preprocessing stage. 
Our result allows us to significantly reduce the space usage of the 
fastest previously known static and incremental $d$-dimensional data 
structures, $d\geq 3$, 
 at a cost of  increasing  the query time by a negligible 
$O(\log \log n)$ factor.
\end{abstract}

\section{Introduction}
The range reporting problem is to store a set of $d$-dimensional
points $P$ in a  data structure, so that for a query rectangle $Q$
all points in $Q\cap P$ can be reported. In this paper we
significantly improve the space usage and pre-processing time of the
fastest previously known static and semi-dynamic data structures for
orthogonal range reporting  with only a negligible increase
in the query time.

The range reporting is extensively studied at least since 1970s; the
history of this problem is rich with different trade-offs between
query time and space usage.  Static range reporting queries can be
answered in $O(\log^d n + k)$ time and $O(n\log^{d-1} n)$ space using
range trees~\cite{B80} known since 1980; here and further $n$ denotes the 
number  of points in  $P$ and   $k$ denotes
the number of points from $P$ in the query rectangle.  The query time can 
be reduced to $O(\log^{d-1} n + k)$ time by applying the fractional
cascading technique of Chazelle and Guibas~\cite{CG86} designed in
1985.  The space usage was further improved by Chazelle~\cite{Ch88}.
In 90s, Subramanian and Ramaswamy~\cite{SR95} and Bozanis, Kitsios,
Makris, and Tsakalidis~\cite{BKMT97} showed that $d$-dimensional
queries can be answered in $\tO(\log^{d-2} n + k)$ time\footnote{We
  define $\tO(f(n))=O(f(n)\log^c(f(n)))$ for a constant $c$.} at a cost
 of higher space
usage: their data structures use $O(n\log^{d-1} n)$ and $O(n\log^{d}
n)$ space respectively.  Alstrup, Brodal, and Rauhe~\cite{ABR00}
designed a data structure that answers queries in $\tO(\log^{d-2} n + k)$
time and uses $O(n\log^{d-2+\eps} n)$ space for an arbitrary constant
$\eps>0$.  Nekrich~\cite{N07} reduced the query time by
$\tO(\log n)$ factor and presented a data structure that answers
queries in $O(\log^{d-3} n/( \log \log n)^{d-5} + k)$ time for $d>3$.
Unfortunately, the data structure of~\cite{N07} uses
$O(n\log^{d+1+\eps} n)$ space. Recently, Afshani~\cite{A08}
reduced the space usage to $O(n\log^{d+\eps} n)$; however his data
structure uses randomization (during the preprocessing stage).  
In this paper we present a data
structure that matches the space efficiency of~\cite{ABR00} at a cost
of increasing the query time by a negligible $O(\log \log n)$ factor:
our data structure supports queries in $O(\log^{d-3} n /(\log \log
n)^{d-6} + k)$ time and uses $O(n\log^{d-2+\eps} n)$ space for $d>3$. 
See Table~\ref{table:stat} for a more precise comparison of different
results.

\tolerance=1000
Our result for $d$-dimensional range reporting is obtained as a
corollary of a three-dimensional data structure that supports
queries in $O(\log \log U + (\log \log n)^3 + k)$ time and uses
$O(n\log^{1+\eps} n)$ space, where  $U$ is the size of the universe, 
i.e. all point coordinates are positive integers bounded by $U$. 
Our three-dimensional data structure is
to be compared with the data structure of~\cite{ABR00} that also uses
$O(n\log^{1+\eps} n)$ space but answers queries in $O(\log n + k)$
time and the data structure of~\cite{N07} that answers queries in
$O(\log \log U + (\log \log n)^3 + k)$ time but needs $O(n\log^4 n)$
space. See Table~\ref{table:3dim} 
for a more extensive comparison
with previous results.  A corollary of our result is an efficient
semi-dynamic data structure that supports three-dimensional queries in
$\tO(\log n + k)$ time and insertions in $O(\log^5 n)$ time. 
Thus we improve the space usage and the update time of fastest previously known  semi-dynamic data structure~\cite{N07}
that supports insertions in $O(\log^8 n)$ time.

\begin{table}[t]
\centering
{
\fontsize{11}{12}
  \selectfont
\begin{tabular}{|l|c|c|} \hline
Source & Query Time & Space \\ \hline
\cite{B80}   & $O(\log^d n +k)$     & $O(n\log^{d-1} n)$ \\ 
\cite{CG86}   & $O(\log^{d-1} n +k)$ & $O(n\log^{d-1} n)$ \\ 
\cite{Ch88} & $O(\log^{d-1} n +k)$ & $O(n\log^{d-2+\eps} n)$ \\ 
\cite{SR95} & $O(\log^{d-2} n\log^{**} n  +k)$ & $O(n\log^{d-1} n)$ \\ 
\cite{BKMT97}  & $O(\log^{d-2} n +k)$ & $O(n\log^{d} n)$ \\ 
\cite{ABR00}  & $O(\log^{d-2} n/(\log \log n)^{d-3} +k)$ & $O(n\log^{d-2+\eps} n)$ \\ 
\cite{N07} & $O(\log^{d-3} n/(\log \log n)^{d-5} +k)$ & $O(n\log^{d+1+\eps} n)$ \\
\cite{A08}$\dagger$ & $O(\log^{d-3} n/(\log \log n)^{d-5} +k)$ & $O(n\log^{d+\eps} n)$ \\ 
This paper & $O(\log^{d-3} n/(\log \log n)^{d-6} +k)$ & $O(n\log^{d-2+\eps} n)$ \\
\hline 
\end{tabular}
}
\caption{\label{table:stat}
Data structures in $d>3$ dimensions; $\dagger$ indicates that a data 
structure is randomized. We define $\log^*(n)=\min\{\,t\,|\,\log^{(t)}n\leq 1\,\}$ and $\log^{**}n=\min\{\,t\,|\, \log^{*(t)} n\leq 1 \,\}$ where $\log^{*(t)}n$ denotes computing $\log^*$ $t$ times.}
\end{table}  

\tolerance=1500 
If we are ready to pay penalties for each point in the
answer, the space usage can be further reduced: we describe a data
structure that uses $O(n\log^{d-2} n (\log \log n)^3)$ space and
answers queries in $O(\log^{d-3} n (\log \log n)^3 + k\log \log n)$
time. We can also use this data structure to answer emptiness queries
(to determine whether query rectangle $Q$ contains points from $P$)
and one-reporting queries (i.e. to report an \emph{arbitrary} point
from $P\cap Q$ if $P\cap Q\not=\emptyset$).  This is an $\tO(\log n)$
factor improvement in query time over the data structure
of Alstrup et.\ al.~\cite{ABR00}.  Other similar data structures are either slower or
require higher penalties for each point in the answer.


 \begin{table}[tb]
 \centering
 {
 \fontsize{11}{12}
   \selectfont
 \begin{tabular}{|l|c|c|} \hline
 Source & Query Time & Space \\ \hline
 \cite{Ch88} & $O(\log^{2} n +k)$ & $O(n\log^{1+\eps} n)$ \\ 
 \cite{SR95} & $O(\log n\log^{**} n  +k)$ & $O(n\log^{2} n)$ \\ 
 \cite{BKMT97}  & $O(\log n +k)$ & $O(n\log^{3} n)$ \\ 
 \cite{ABR00}  & $O(\log n +k)$ & $O(n\log^{1+\eps} n)$ \\ 
 \cite{N07} & $O(\log \log U +(\log \log n)^2  +k)$ & $O(n\log^{4+\eps} n)$ \\
 \cite{A08}$\dagger$ & $O(\log \log U +(\log \log n)^2 +k)$ & $O(n\log^{3} n)$ \\ 
 This paper & $O(\log \log U +(\log \log n)^3 +k)$ & $O(n\log^{1+\eps} n)$ \\
 \hline 
 \end{tabular}
 }
 \caption{\label{table:3dim}
 Three-dimensional data structures; $\dagger$ indicates that a data 
 structure is randomized.}
 \end{table}  

Throughout this paper, $\eps$ denotes an arbitrarily small constant,
and $U$ denotes the size 
of the universe. If each point in
the answer can be output in constant time, we will sometimes say that
the query time is $O(f(n))$ (instead of $O(f(n)+k)$).
We let $[a,b]$ denote the set of integers $\{i| a\leq i \leq b\}$;
The intervals $\halfrightsect{a}{b}$ and $\halfleftsect{a}{b}$ denote 
the same set as $[a,b]$ but without $a$ (resp.\ without $b$). 
We denote by  $[b]$ the set $[1,b]$.
 
In section~\ref{sec:3d1} we describe a space efficient data structure
for three-dimensional range reporting on a $[U]\times [U]\times [U]$ grid, 
i.e. in the case when all point coordinates belong to $[U]$.  In
section~\ref{sec:3d2} we describe a variant of our data structure that
uses less space but needs $O(\log \log n)$ time to output each point
in the answer. All results of this paper are valid in the word RAM
 computation model.



\section{Preliminaries}
\label{sec:prelim}
We use the same notation as in~\cite{VV96} to denote the special cases 
of three-dimensional range reporting queries: 
a product of three half-open intervals will be called a (1,1,1)-sided query;
a product of a closed interval and two half-open intervals 
will be called a (2,1,1)-sided query; a product of two closed intervals 
and one half-open interval (resp.\ three closed intervals) will 
be called a (2,2,1)-sided (resp.\ (2,2,2)-sided) query. 
Clearly (1,1,1)-sided queries 
are equivalent to  dominance reporting queries, and (2,2,2)-sided query 
is the general three-dimensional query.
The following  transformation is described 
in e.g.~\cite{VV96} and~\cite{SR95}. 
\begin{lemma}\label{lemma:transf}
Let $1\leq a_i\leq b_i \leq 2$ for $i=1,2,3$. 
A data structure that answers  $(a_1,a_2,a_3)$ queries 
in $O(q(n))$ time,  uses $O(s(n))$ space, and can be constructed in  
$O(c(n))$ time  
can be transformed into a data structure that answers 
 $(b_1,b_2,b_3)$ queries in $O(q(n))$ time, uses $O(s(n)\log^t n)$ space 
and can be constructed in $O(c(n)\log^t n)$ time 
for $t=(b_1-a_1)+(b_2-a_2)+(b_3-a_3)$.
\end{lemma}
 
We say that a set $P$ is on a grid of size $n$ if all coordinates of all 
points in $P$ belong to an interval $[n]$. 
We will  need the following folklore result:
\begin{lemma}\label{lemma:folk}
There exists a $O(n^{1+\eps})$ space data structure that supports 
 range reporting queries on a $d$-dimensional grid 
of size $n$ for any constant $d$ 
in $O(k)$ time.
\end{lemma}
\begin{proof}
One dimensional range reporting queries on the 
$[n]\times [n]\times [n]$ grid  can be answered 
in $O(k)$ time using a trie with node degree $n^{\eps}$. 
Using range trees~\cite{B80} with node degree $\rho$ we can transform 
a $d$-dimensional $O(s(n))$ space data structure into a $(d+1)$-dimensional 
data structure that uses $O(s(n)h(n)\cdot \rho)$ space and
 answers range reporting 
queries in $O(q(n)h(n))$ time, where $h(n)=\log n/\log \rho$ is the height 
of the range tree. Since $\rho=n^{\eps}$, $h(n)=O(1)$. 
Hence, the query time does not depend on dimension and the space usage 
increases by a factor $O(n^{\eps})$ with each dimension.
\end{proof}
We use  Lemma~\ref{lemma:folk} to obtain a data structure that 
 supports queries that are a product of a $(d-1)$-dimensional query on a 
universe of size $n^{1-\eps}$ 
and  a half-open interval. We will show in the next Lemma that such 
queries can be answered in $O(n)$ space and $O(1)$ time. 
\begin{lemma}\label{lemma:muchspace}
There exists a $O(n)$ space data structure that supports 
 range reporting queries of the form $Q'\times\halfrightsect{-\infty}{x}$ 
in $O(k)$ time, 
where $Q'$ is a $(d-1)$-dimensional query on $[U_1]\times[ U_2]\times\ldots\times
 [U_{d-1}]$ and $U_1\cdot U_2 \cdot\ldots\cdot U_{d-1}= O(n^{1-\eps}) $.
\end{lemma}
\begin{proof}
There are $O(n^{1-\eps})$ possible projections of points onto the 
first $d-1$ coordinates. Let $\min(p_1,\ldots,p_{d-1})$ denote the 
point with minimal $d$-th coordinate among all points whose first 
$d-1$ coordinates equal to $p_1,p_2,\ldots,p_{d-1}$. 
We store  points $\min(p_1,\ldots,p_{d-1})$ for all 
$p_1\in [U_1]$,$p_2\in [U_2]$,$\ldots$,$p_{d-1}\in [U_{d-1}]$ in a 
data structure $M$. Since $M$ contains $O(n^{1-\eps})$ points, 
we can use Lemma~\ref{lemma:folk} and implement $M$ in $O(n)$ space. 
For all possible  $p_1\in [U_1]$,$p_2\in [U_2]$,$\ldots$,
$p_{d-1}\in [U_{d-1}]$ 
we also store a list $L(p_1,\ldots, p_{d-1})$ 
of points whose first $d-1$ coordinates are 
$p_1,\ldots,p_{d-1}$; points in  $L(p_1,\ldots, p_{d-1})$ are  
sorted by their  $d$-th coordinates. 
Given a query $Q=Q'\times\halfrightsect{-\infty}{x}$, we  first  
answer $Q$ using the data structure $M$. Since $M$ contains $O(n^{1-\eps})$ 
points, we can find all points in $M\cap Q$ in $O(|M\cap Q|)$ time. 
Then, for every point  
$p=(p_1,\ldots,p_{d-1},p_d)$ found 
with help of $M$, we traverse the corresponding list $L(p_1,\ldots,p_{d-1})$ 
and report all points in this list whose last coordinate does not exceed 
$x$. 
\end{proof}

In several places of our proofs we will use the 
\emph{reduction to rank space} technique~\cite{GBT84,Ch88}. 
This technique allows us to replace coordinates of a point by its rank. 
Let $P_x$, $P_y$, and $P_z$ be the sets of $x$, $y$-, and $z$-coordinates 
of points from $P$. For a point $p=(p_x,p_y,p_z)$, let 
 $p'=(\rank(p_x,P_x),\rank(p_y,P_y),\rank(p_z,P_z))$, 
where $\rank(e,S)$ is defined as the number of elements in $S$ that 
are smaller than or equal to $e$. 
A point $p$ belongs to an interval $[a,b]\times [c,d] \times [e,f]$ 
if and only if a point $p'$ belongs to an interval 
$[a',b']\times [c',d'] \times [e',f']$ where 
$a'=\successor(a,P_x)$, $b'=\pred(b,P_x)$, 
$c'=\successor(c,P_y)$, $d'=\pred(d,P_y)$, 
$e'=\successor(e,P_z)$, $f'=\pred(f,P_z)$, 
and $\successor(e,S)$  ($\pred(e,S)$) denotes 
the smallest (largest) element in S that is greater (smaller) 
than or equal to $e$. 
Reduction to rank space can be used to reduce range reporting queries to 
range reporting on the $[n]\times [n]\times [n]$ grid: Suppose  we can 
find $\pred(e,s)$ and $\successor(e,S)$ for any $e$, where $S$ is $P_x$, $P_y$, 
or $P_z$, in time $f(n)$. Suppose that range reporting queries on 
$[n]\times [n]\times [n]$ grid can be answered in time $O(g(n)+k)$. 
Then we can answer range reporting queries in $O(f(n)+g(n)+k)$ time. 
Following~\cite{ABR00}, we can also use the reduction to rank space 
 technique to reduce the space usage: if a data structure contains $m$ elements, reduction to rank space 
allows us to store each element in $O(\log m)$ bits.

\section{Space Efficient Three-Dimensional Data Structure}
\label{sec:3d1}
In this section we describe a data structure that supports three-dimensional 
range reporting queries in $O((\log \log n)^3 +\log \log U + k)$ time
where $U$ is the universe size and uses $O(n\log^{1+\eps} n)$ space.  
Our data structure  combines  the recursive divide-and-conquer approach 
introduced in~\cite{ABR00}, the result of Lemma~\ref{lemma:muchspace}, 
 and the transformation of 
$(a_1,a_2,a_3)$-queries into $(b_1,b_2,b_3)$-queries described in 
Lemma~\ref{lemma:transf}. We start with a description of a space efficient 
modification of the data structure for (1,1,1)-sided queries on the 
$[n]\times [n]\times [n]$ grid. 
Then, we obtain data structures for $(2,1,1)$-sided and $(2,2,1)$-sided 
queries on the $[n]\times [n]\times [n]$ grid using the recursive 
divide-and-conquer and 
Lemma~\ref{lemma:muchspace}. Finally, we obtain the data structure 
that supports arbitrary orthogonal queries on the 
$[n]\times [n]\times [n]$ grid using Lemma~\ref{lemma:transf}.
Reduction to rank space technique
described in  section~\ref{sec:prelim} allows us to 
transform a data structure 
on the $[n]\times [n]\times [n]$  grid   into a data structure on the 
$[U]\times [U]\times [U]$ grid , so that the query time increases by an 
additive term  
$O(\log \log U)$ and the space usage is not increased.
\begin{lemma}\label{lemma:tbound}~\cite{N07}
Given a set of three-dimensional points $P$ and a parameter $t$, 
we can construct in $O(n\log^3 n)$ time  a $O(n)$ space data structure  
$T$ that supports the following queries on a grid of size $n$: \\
(i) for a given query point $q$, $T$ determines in $O((\log \log n)^2)$ time 
whether $q$ is dominated by at most $t$ points of $P$\\
(ii) if  $q$ is dominated by at most 
$t$ points from $P$, $T$ outputs in $O(t + (\log \log n)^2)$ 
time a list $L$ of $O(t)$ points such  that $L$ 
contains all points of $P$ that dominate $q$.
\end{lemma}

As described in~\cite{N07}, Lemma~\ref{lemma:tbound} allows us to answer 
(1,1,1)-sided queries in $O((\log \log n)^2)$ time 
and $O(n\log n)$ space.  We can reduce the space usage to 
$O(n\log \log n)$ using an idea that is also used in~\cite{A08}.
\begin{lemma}\label{lemma:domin}
There exists a data structure that answers (1,1,1)-sided queries on 
$[n]\times [n]\times [n]$ grid  
in $O((\log \log n)^2 + k)$ time, uses $O(n\log \log n)$ space, 
and can be constructed in $O(n\log^3 n \log \log n)$ time. 
\end{lemma}
\begin{proof}
For each parameter $t=2^{2i}$, $i=i_{\min},i_{\min+1},\ldots, 
\log \log n/2$,  $i_{\min}=2 \log \log \log n$,
we construct a data structure $T_i$ of Lemma~\ref{lemma:tbound}. 
Given a query point $q$, we examine data structures $T_i$, 
$i= i_{\min},i_{\min+1},\ldots, \log \log n/2$ until  $q$ is dominated by at 
most $2^{2i}$ points of $P$ or the last data structure $T_i$ is examined. 
Thus we  identify the index $l$, such that $q$ is
 dominated by more than $2^{2l}$ and less than $2^{2l+2}$ points 
or determine that $q$ is dominated by at least $\log n$ points. 
If $l=i_{\min}$, then $q$ is dominated by $O((\log \log n)^2)$ points. 
We can generate in $O((\log \log n)^2)$ time a list 
$L$ of $O((\log \log n)^2)$ points that contains all points dominating $q$. 
Then, we examine all  points in $L$ and output all points that dominate 
$q$ in $O((\log \log n)^2)$ time.  
If $\log \log n/2 > l  >i_{\min}$, we can  examine data structures 
$T_{i_{\min}}$, $T_{i_{\min}+1}$,$\ldots$, $T_l$ in  
$O((l-i_{\min})(\log \log n)^2)$ time. 
Then, we  generate the list $L$ that contains all points that dominate 
$q$ in $O(2^{2l})$ time. We can process $L$ and output all $k$ points that 
dominate $q$ in $O(2^{2l})$ time. Since $k> 2^{2l-2}$, 
$k=\Omega(2^{2l})$ and $k=\Omega((l-i_{\min})\cdot(\log \log n)^2)$. 
Hence, the  query is answered in $O(k)$ time.  
If $l=\log \log n/2$, then $q$ is dominated by $\Omega(\log n)$ points. 
in this case we can use 
a linear space data structure with $O(\log n)$ query time, e.g. the data 
structure of Chazelle and Edelsbrunner~\cite{ChE87}, to answer the 
query in $O(\log n + k)=O(k)$ time.

Since each data structure $T_i$ uses linear space, the space usage 
of the described data structure is $O(n\log \log n)$.
\end{proof}

\begin{lemma}\label{lemma:two11}
There exists a data structure that answers (2,1,1)-sided queries on 
$[n]\times [n]\times [n]$ grid 
in $O( (\log \log n)^3 + k)$ time, uses 
$O(n\log^{\eps} n)$ space, 
and can be constructed in $O(n\log^{3} n \log \log n)$ time. 
\end{lemma}
\begin{proof}
  We divide the grid into $x$-slices $X_i=[x_{i-1},x_i]\times [n]\times
  [n]$ and $y$-slices $Y_j=[n]\times [y_{j-1},y_j]\times [n]$, so that each
  $x$-slice contains $n^{1/2+\gamma}$ points and each $y$-slice
  contains $n^{1/2+\gamma}$ points; the value of a constant $\gamma$
  will be specified below.  The cell $C_{ij}$ is the intersection of
  the $i$-th $x$-slice and the $j$-th $y$-slice, $C_{ij}=X_i\cap Y_j$.
  The data structure $D_t$ contains a point $(i,j,z)$ for each point
  $(x,y,z)\in P\cap C_{ij}$. Since the first two coordinates of points
  in $D_t$ are bounded by $n^{1/2-\gamma}$, $D_t$ uses $O(n)$ space
  and supports (2,1,1)-sided queries in constant time by
  Lemma~\ref{lemma:muchspace}.  For each $x$-slice $X_i$ there are two
  data structures that support two types of (1,1,1)-sided queries,
  open in $+x$ and in $-x$ directions.  For each $y$-slice $Y_j$,
  there is a data structure that supports $(1,1,1)$-sided queries open
  in $+y$ direction.  For each $y$-slice $Y_j$ and for each $x$-slice
  $X_i$ there are recursively defined data structures.
  Recursive subdivision stops when the number of elements in a data 
  structure is smaller than a predefined constant. Hence, the number of 
  recursion levels is $v\log \log n$ for $v=\log_{\frac{2}{1+2\gamma}}2$.

Essentially we apply the idea of~\cite{ABR00} to three-dimensional $(2,1,1)$-sided queries. If a query spans more than one $x$-slab and more than one 
$y$-slab, then it can be answered by answering two $(1,1,1)$-sided queries, 
one special $(2,1,1)$-sided query that can be processed using the 
technique of Lemma~\ref{lemma:muchspace}, and one $(2,1,1)$-sided 
query to a data structure with $n^{1/2+\gamma}$ points. If a query 
is contained in a slab, then it can be answered by a data structure 
that contains $n^{1/2+\gamma}$ points. We will show below that 
the query time is $O((\log \log n)^3)$.  Each point is stored in 
$O(2^i)$ data structures on recursion level $i$, but space usage 
can be reduced because the number of points in data structures 
quickly decreases with the recursion level. 
We will show below that every point in a data structure on recursion level 
$i$ can be stored with approximately 
$(\log n/2^i)\log^{\eps'}n$ bits for an arbitrarily 
small $\eps'$.  

{\bf Query Time.} Given a query $Q=[a,b]\times\halfleftsect{-\infty}{c}\times
\halfleftsect{-\infty}{d}$ we identify the indices $i_1$, $i_2$, and $j_1$ 
such that projections of all cells $C_{ij}$, $i_1<i <i_2$, $j<j_1$, are
entirely contained in $[a,b]\times\halfleftsect{-\infty}{c}$.  Let
$a_0=x_{i_1}$, $b_0=x_{i_2-1}$, and $c_0=y_{j_1-1}$.  The query $Q$
can be represented as $Q =Q_1\cup Q_2\cup Q_3\cup Q_4$, where
$Q_1=[a_0,b_0]\times \halfleftsect{-\infty}{c_0}\times
\halfleftsect{-\infty}{d}$, $Q_2= \halfrightsect{a}{a_0} \times
\halfleftsect{-\infty}{c}\times \halfleftsect{-\infty}{d}$,
$Q_3=\halfleftsect{b_0}{b}\times \halfleftsect{-\infty}{c} \times
\halfleftsect{-\infty}{d}$, and
$Q_4=[a_0,b_0]\times\halfleftsect{c_0}{c}\times\halfleftsect{-\infty}{d}$.
See Fig.~\ref{fig:exam211} for an example. 
Query $Q_1$ can be answered using $D_t$. Queries $Q_2$ and $Q_3$ can
be represented as $Q_2 = (\halfrightsect{-\infty}{a_0} \times
\halfleftsect{-\infty}{c}\times \halfleftsect{-\infty}{d})\cap
X_{i_1}$ and $Q_3= (\halfleftsect{-\infty}{b}\times
\halfleftsect{-\infty}{c} \times \halfleftsect{-\infty}{d})\cap
X_{i_2}$; hence, $Q_2$ and $Q_3$ are equivalent to $(1,1,1)$-sided
queries on $x$-slices $X_{i_1}$ and $X_{i_2}$.  The query $Q_4$ can be
answered by a recursively defined data structure for the $y$-slice
$Y_{j_1}$ because $Q_4=
([a_0,b_0]\times\halfleftsect{-\infty}{c}\times\halfleftsect{-\infty}{d})
\cap Y_{j_1}$.  If $i_1=i_2$ and the query $Q$ is contained in one
$x$-slice, then $Q$ is processed by a recursively defined data
structure for the corresponding $x$-slice.  Thus a query is reduced 
to one special case that can be processed in constant time, two 
$(1,1,1)$-sided  queries, and one (2,1,1)-sided  
query answered by a data structure 
that contains $n^{1/2+\gamma}$ elements.

\begin{figure*}[t]
\centering
\includegraphics[width=80mm]{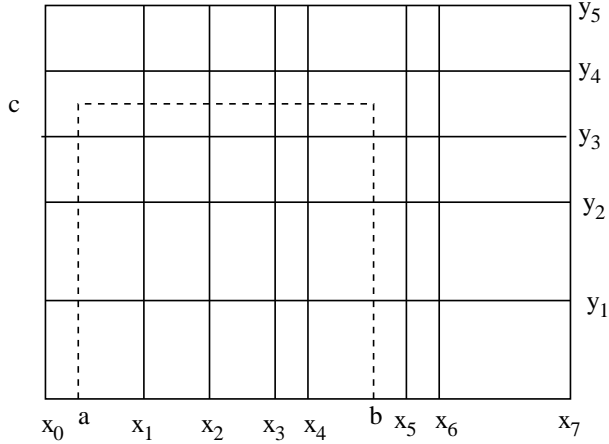}
\caption{\label{fig:exam211} Example of a $(2,1,1)$-sided query projected 
onto the $xy$-plane. $i_1=1$, $i_2= 5$, $j_1=4 $ and 
$a_0=x_1$, $b_0=x_4$, $c_0=y_3$.}
\end{figure*}

Queries $Q_2$ and $Q_3$ can be answered in $O((\log \log n)^2)$ time,
 the query $Q_1$ can be answered in constant time. The query
$Q_4$ is answered by a recursively defined data structure that
contains $O(n^{1/2+\gamma})$ elements.  
If $i_1=i_2$ or $j_1=1$, i.e. if $Q$ is entirely contained in one $x$-slice or one 
$y$-slice, then the query is answered by a data structure for the corresponding slice that contains $O(n^{1/2+\gamma})$ elements. Hence, the query time 
$q(n) = O((\log \log n)^2) + q(n^{1/2+\gamma})$ and $q(n)=O((\log \log
n)^3)$.

{\bf Space Usage.} 
The data structure consists of $O(\log \log n)$ recursion levels. The total number 
of points in all data structures on the $i$-th recursion level is $2^i n$. 
Hence all data structures on the $i$-th recursion level require $O(2^i n \log n)$
bits of space. 
The space usage can be reduced by applying the reduction to rank space
 technique~\cite{GBT84,Ch88}. 
As explained in section~\ref{sec:prelim}, reduction to rank space allows us
 to replace  point coordinates
by their  ranks. Hence, if we use this technique with a data structure 
that contains $m$ elements, each point can be specified with $O(\log m)$ bits.
Thus, we can reduce the space usage by replacing point coordinates 
by their ranks on certain recursion levels. 
 
We apply reduction to rank space on every $\delta\log \log n$-th recursion level for $\delta=\eps/3$. 
Let $V$ be an arbitrary data structure on recursion level
 $r=s\delta\log \log n-1$  for 
$1\leq s \leq (1/\delta) \log_{\frac{2}{1+2\gamma}} 2 $. 
Let $W$ be the set of points that belong to 
 an $x$-slice or a $y$-slice of $V$.
We store a dictionary that enables us to find for each point 
$p=(p_x,p_y,p_z)$ from $W$ 
a point $p'=(p_x',p_y',p_z')$ where $p'_x=\rank(p_x,W_x)$,
$p'_y=\rank(p_y,W_y)$, $p'_z=\rank(p_z,W_z)$, and $W_x$,$W_y$, and $W_z$ 
are the sets of $x$-, $y$-, and $z$-coordinates of all points in $W$.
Let $W'$ be the set of all points $p'$. 
Conversely there is also a dictionary that enables us to find for a point 
$p'\in W'$ the corresponding $p\in W$.  
The data structure that answers queries on $W$ stores points 
in the rank space of $W$. 
In general, all data structures on recursion levels 
$r, r+1,\ldots, r+\delta\log \log n-1$ obtained by subdivision 
of $W$ store points in rank space of $W$. 
That is, point coordinates in all those data structures are integers 
bounded by $|W|$.
If such a data structure $R$ is used to answer  a query $Q$, then 
for each point $p_R\in R\cap Q$, we must find the corresponding point 
$p\in P$. Since range reduction was applied $O(1)$ time, we can find 
for any $p_R\in R$ the corresponding $p\in P$ in $O(1)$ time.

Each  data structure on level $r=s\delta\log \log n$  for 
$0\leq s \leq (1/\delta)  v  $ and 
$v= \frac{1}{\log (2/(1+2\gamma)) }$ 
contains $O(n^{l})$ elements for $l=(1/2+\gamma)^{r}$. 
Hence an arbitrary element of a data structure on level $r$ can be specified 
with  $l\cdot \log n$ bits.  
The total number of elements in all 
data structures on the $r$-th level is $n2^{r}$. 
Hence all elements in all data structures on the $r$-th recursion level need 
$O(n2^r((\frac{1+2\gamma}{2})^r)\log n \log \log n)$ bits.

We choose $\gamma$ so that $(1+2\gamma)\leq 2^{\delta/2}$. 
Then $v  
= \frac{1}{1-\log_2 (1+2\gamma)}\geq\frac{1}{1-\delta/2}$ and $(1+2\gamma) \leq 2^{\delta/2} \leq 2^{\delta-\delta^2/2}
\leq 2^{\delta/v}=2^{\eps/3v}$. 
Since $r\leq v\log \log n$,  
$(1+2\gamma)^r\leq 2^{(\eps/3)\log \log n}\leq \log^{\eps/3} n$. 
Therefore all data structures on level $r$ use 
$\log^{\eps/3} n \cdot O(n\log n \log \log n)=
O(n\log^{1+2\eps/3}n)$ bits of space or $O(n\log^{2\eps/3}n)$ words 
of $\log n$ bits.
The number of elements in all data structures on levels 
$r+1,r+2,\ldots$ increases by a factor two in each level. 
Hence, the total space (measured in words) needed for all data structures 
on all levels $q$, $r\leq q < r+ \delta\log \log n$, 
is $(\sum_{f=1}^{\delta\log \log n -1} 2^f) O(n\log^{2\eps/3}n)=
O(n2^{\delta\log\log n} n\log^{2\eps/3}n) =O(n\log^{\eps} n)$ 
because  $\delta\leq \eps/3$ and  $2^{\delta\log\log n} \leq 
\log^{\eps/3} n$. 
Thus all data structures in a group of $\delta \log \log n$ 
consecutive recursion levels use $O(n\log^{\eps} n)$ words of space.
Since there are $(1/\delta)v=O(1)$ such groups of levels, the total space 
usage is $O(n\log^{\eps} n)$.

{\bf Construction Time.} 
The data structure on level 0 (the topmost recursion level) 
can be constructed in $O(n\log^3 n \log  \log n)$ time. 
The total number of elements in all data structures on level $s$ 
is $2^sn\log \log n$. But each data structure on the $r$-th recursion  level 
contains at most 
$n_r=n^l$ elements and can be constructed in $O(l^3\cdot n_r\log^{3} n \log \log n)$
 time where $l=(1+2\gamma)^r/2^{r}$. 
Hence, all data structure on the $r$-th recursion level can be constructed 
in $O((2^rl^3) n\log^3 n \log \log n)= O(((1+2\gamma)^{3r}/2^{2r}) n\log^3 n \log \log n)$ time. 
We can assume that $\eps <1$. Since we chose $\gamma$ so that 
 $(1+2\gamma) \leq  2^{\eps/6}$, $(1+2\gamma)^3 < 2$; hence,  
$(1+2\gamma)^{3r}/2^{2r}\leq 1/2^r$.
Then, all data structure on the $r$-th recursion level can be constructed 
in $O((1/2^r) n \log^3 n \log \log n)$ time. 
Summing up by all $r$, we see that all recursive data structures can be
 constructed in $O(n\log^3 n \log \log n)$ time.
\end{proof}


\begin{lemma}\label{lemma:two21}
There exists a data structure that answers (2,2,1)-sided queries on 
$[n]\times [n]\times [n]$ grid
in $O( (\log \log n)^3 + k)$ time, 
uses $O(n\log^{\eps} n)$ space, 
and can be constructed in $O(n\log^{3} n \log \log n)$ time. 
\end{lemma} 
\begin{proof}
The proof technique is the same as in Lemma~\ref{lemma:two11}. 
The grid is divided into $x$-slices $X_i=[x_{i-1},x_i]\times n\times n$ 
and $y$-slices $Y_j=n\times [y_{j-1},y_j]\times n$  in the same way 
as in the proof of Lemma~\ref{lemma:two11}. 
Each $x$-slice $X_i$ supports $(2,1,1)$-sided queries open in $+x$ and $-x$ 
direction; each $y$-slice $Y_j$ supports $(2,1,1)$-sided queries open in 
$+y$ and $-y$ direction. 
All points are also stored in a data structure $D_t$ that 
contains a point $(i,j,z)$ for each point $(x,y,z)\in P \cap C_{ij}$. 
For every $x$-slice and $y$-slice there is a recursively defined data 
structure. The reduction to rank space technique is applied 
on every $\delta \log \log n$-th level in the same way as in the
 Lemma~\ref{lemma:two11}. 

Given a query $Q=[a,b]\times [c,d]\times \halfleftsect{-\infty}{e}$ 
we identify indices $i_1,i_2,j_1,j_2$ such that all 
cells $C_{ij}$, $i_1< i <i_2$ and $j_1<j<j_2$ are entirely contained 
in $Q$. Then $Q$ can be represented as a union of a query 
$Q_1=[a_0,b_0]\times [c_0,d_0]\times\halfleftsect{-\infty}{e}$ 
and four $(2,1,1)$-sided  queries 
$Q_2= \halfrightsect{a}{a_0} \times [c,d]\times \halfleftsect{-\infty}{e}$, 
$Q_3= \halfleftsect{b_0}{b}\times [c,d]\times\halfleftsect{-\infty}{e}$, 
$Q_4=[a_0,b_0]\times \halfrightsect{c}{c_0}\times \halfleftsect{-\infty}{e}$, 
and 
$Q_5=[a_0,b_0]\times \halfleftsect{d_0}{d} \times\halfleftsect{-\infty}{e}$, 
where $a_0=x_{i_1}$, $b_0=x_{i_2-1}$,
$c_0 =y_{j_1}$, and $d_0=y_{j_2-1}$. See Fig.~\ref{fig:exam221} for an
 example.  
The query $Q_1$ can be answered in constant time, and queries 
$Q_i$, $1<i\leq 5$, can be answered using the corresponding 
$x$- and $y$-slices. Since queries $Q_i$, $1 < i \leq 5$, are 
equivalent to (2,1,1)-sided queries each of those 
queries can be answered in $O((\log \log n)^3 + k)$ time. 

If the query $Q$ is entirely contained in one $x$-slice or one $y$-slice, 
then $Q$ is processed by a data structure for the corresponding 
$x$-slice resp.\ $y$-slice. Since the data structure consists of 
at most $v \log \log n$ recursion levels, the query can be transferred 
to a data structure for an $x$- or $y$-slice at most 
$v\log \log n$ times for $v= \frac{1}{\log (2/(1+2\gamma))}$. 
Hence, the total query time 
is $O(\log \log n + (\log \log n)^3 + k)= O((\log \log n)^3 + k)$.
The space usage and construction time are estimated in  the same way 
as in Lemma~\ref{lemma:two11}.
\begin{figure*}[t]
\centering
\includegraphics[width=80mm]{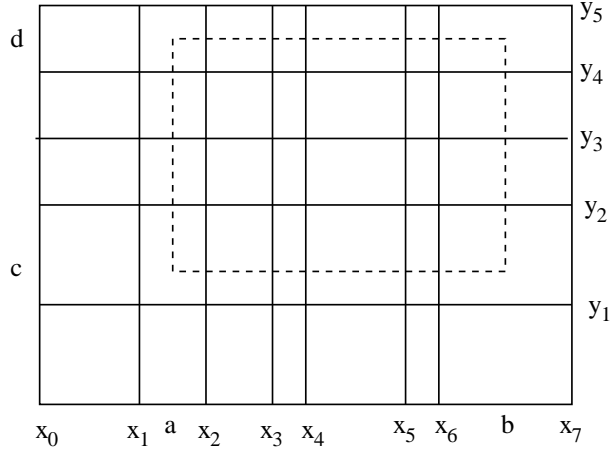}
\caption{\label{fig:exam221} Example of a (2,2,1)-sided  query projected 
onto the $xy$-plane. $i_1=2$, $i_2=7$, $j_1=2$, and $j_2=5$.}
\end{figure*}
\end{proof}

\begin{theorem}\label{theor:threedim1}
There exists a data structure that answers three-dimensional orthogonal 
range reporting queries on the $[U]\times [U]\times [U]$ grid  
in $O(\log \log U +(\log \log n)^3 + k)$ time, uses $O(n\log^{1+\eps} n)$ space, 
and can be constructed in $O(n\log^{4} n \log \log n)$ time. 
\end{theorem}
\begin{proof}
The result for the $[n]\times [n]\times [n]$ grid 
 directly follows from Lemma~\ref{lemma:two21} and 
Lemma~\ref{lemma:transf}. We can obtain the result for the $[U]\times [U]\times [U]$ grid by applying the reduction to rank space
 technique~\cite{GBT84,Ch88}: 
We can use the van Emde Boas data structure~\cite{E77} to find $\pred(e,S)$ and $\successor(e,S)$ for any $e\in [U]$ in 
$O(\log \log U)$ time, where $S\subset [U]$ is $P_x$, $P_y$, or $P_z$. 
Hence, the query time is increased by an additive term $O(\log \log U)$ 
and the space usage remains unchanged. 
\end{proof}

Furthermore, we also obtain the result for $d$-dimensional 
range reporting, $d\geq 3$.
\begin{corollary}
There exists a data structure that answers $d$-dimensional orthogonal 
range reporting queries 
in $O(\log^{d-3} n/(\log \log n)^{d-6} + k)$ time, uses 
$O(n\log^{d-2+\eps} n)$ space, 
and can be constructed in $O(n\log^{d+1+\eps} n)$ time. 
\end{corollary}
\begin{proof}
We can obtain a $d$-dimensional data structure from a $(d-1)$-dimensional 
data structure using range trees with node degree $\log^{\eps} n$. 
See e.g.~\cite{ABR00},~\cite{N07} for details.
\end{proof}

\tolerance=2000
Using Theorem~\ref{theor:threedim1} we can reduce the space usage 
and update time of the semi-dynamic data structure for three-dimensional 
range reporting queries. 
\begin{corollary}\label{cor:semidyn}
There exists a data structure that uses $O(n\log^{1+\eps} n)$ space, 
and supports three-dimensional orthogonal range reporting queries  in 
$O(\log n (\log \log n)^2 + k)$ time and insertions in 
$O(\log^{5+\eps} n)$ time. 
\end{corollary}
\begin{proof}
We can obtain the semi-dynamic data structure from the static data 
structure using a variant of the logarithmic method~\cite{B79}.
 A detailed description can be found in~\cite{N07}.
The space usage remains the same, the query time increases by a 
$O(\log n/\log \log n)$ factor, and the amortized insertion 
time is $O(\frac{c(n)}{n}\log^{1+\eps} n)$, where $c(n)$ is the construction 
time of the static data structure.
\end{proof}

The result of Corollary~\ref{cor:semidyn} can be also extended to 
$d>3$ dimensions using range trees. 

\section{Three-Dimensional Emptiness Queries}
\label{sec:3d2}
We can further reduce the space usage of the three-dimensional 
data structure if we allow $O(\log \log n)$ penalties for each point in 
the answer.  
Such a data structure can also be used to answer emptiness and one-reporting 
queries.
As in the previous section, we design space efficient data structures for 
$(2,1,1)$-sided and $(2,2,1)$-sided  queries. The proof is quite similar 
to the data structure of section~\ref{sec:3d1} but some parameters must 
be chosen in a slightly different way.   
\begin{theorem}\label{theor:penalty}
There exists a data structure that answers three-dimensional orthogonal 
range reporting queries on the $[U]\times [U]\times [U]$ grid
in $O(\log \log U + (\log \log n)^3 + k\log \log n)$ time, 
uses $O(n\log n (\log \log n)^3)$ space, 
and can be constructed in time $O(n\log^{4} n \log \log n)$ . 
\end{theorem}
For completeness,  we provide 
the proof of Theorem~\ref{theor:penalty} in the Appendix.
Using the standard range trees and reduction to rank space techniques 
we can obtain a $d$-dimensional data structure for $d>3$
\begin{corollary}
There exists a data structure that answers $d$-dimensional orthogonal 
range reporting queries for $d>3$ 
in $O(\log^{d-3}n (\log \log n)^3 + k\log \log n)$ time, 
uses $O(n\log^{d-2} n (\log \log n)^3)$ space, 
and can be constructed in $O(n\log^{d+1} n \log \log n)$ time. 
\end{corollary}

\section*{Appendix. Proof of Theorem~\ref{theor:penalty} }

\begin{lemma}\label{lemma:empt211}
There exists a data structure that answers (2,1,1)-sided queries 
on the $[n]\times [n]\times [n]$ grid in $O( (\log \log n)^3 + k\log \log n)$ 
time, 
uses $O(n (\log \log  n)^{2})$ space, 
and can be constructed in $O(n\log^{3} n \log \log n)$ time. 
\end{lemma}
\begin{proof}
The data structure consists of the same components 
as the data structure of Lemma~\ref{lemma:two11}. 
But the size of $x$-slices and $y$-slices is reduced, so that 
each $x$-slice and each $y$-slice contains $n^{1/2}\log^p n$ 
points for a constant $p\geq 2$. 
The data structure $D_t$ contains a point $(i,j,z_{\min})$ for 
each cell $C_{ij}=X_i\cap Y_j$, $C_{ij}\cap P\not=\emptyset$, 
such that  $z_{\min}$ is the minimal $z$-coordinate of a point in 
$ C_{ij}\cap P$. 
The data structure $D_t$ can contain up to 
$n/\log^{2p} n$ elements. 
Combining the results of Lemma~\ref{lemma:transf} and Lemma~\ref{lemma:domin},
we can implement $D_t$ in $O((n/\log^{2p} n)\log n \log \log n)= O(n)$ space, 
so that queries are supported in $O((\log \log n)^2 + k)$ time. 
A list $L_{ij}$ contains all points in $C_{ij}$ sorted by their 
$z$-coordinates.  
For each $x$-slice $X_i$, there are two  data structures that 
support $(1,1,1)$-sided queries open in $+x$ and $-x$ direction. 
For each $y$-slice $Y_j$ there is a data structure for  $(1,1,1)$-sided 
queries open in $+y$ direction. 
For each $x$-slice and $y$-slice, there is a recursively defined data 
structure. 
As shown in Proposition~1 of~\cite{N05}, the total number 
of elements in a data structure on the $r$-th recursion level can 
be estimated as $s^r(n)=O(n^{1/2^r}\log^p n \sqrt{\log \log n})$. 
The recursive sub-division stops when a data structure contains no more than
$\log n$ elements. In this case, the data structure is implemented 
using e.g. the data structure of~\cite{ABR00}, so that queries 
are answered in $O(\log \log n)$ time and $O(\log n(\log \log n)^{1+\eps})$ 
space.

In the same way as in Lemma~\ref{lemma:two11}, the query $Q$ can be 
represented as a union of (at most) one (2,1,1)-sided 
query on $D_t$, two (1,1,1)-sided queries on $x$-slices, and 
one (2,1,1)-sided query on a recursively defined data structure 
for a $y$-slice. 
Hence, the query time is $O((\log \log n)^3)$ if we ignore the time 
we need to output  points in the answer.

Unlike the data structure of Lemma~\ref{lemma:two11}, we apply range 
reduction on every recursion level. 
Since the number of elements in a data structure on level $r$ 
is $s^r(n)=O(n^{1/2^r}\log^p n \sqrt{\log \log n})$, every element in 
a data structure on level $r$ can be represented with 
$\log(s^r(n))=O((1/2^r)\log n + \log \log n)$ bits. 
Each data structure on level $r$ uses $O(s^r(n) \log (s^r(n)) \log \log (s^r(n)) )= O(s^r(n)\log(s^r(n))\log \log n)$ bits. 
The total number of elements in all data structures on level $r$ 
is $O(n2^r)$. Hence, all level $r$ data structures need 
$O(n\log n\log \log n + n2^r (\log \log n)^2)$ bits. 
Summing up by all recursion levels, the total space usage is 
$O(n\log n (\log \log n)^2) + \sum_{r=1}^{r_{\max}-1}n2^r(\log \log n)^2$ 
bits. 
The maximum recursion level  $r_{\max}=\log \log n + c_r$ for a
constant $c_r$. Hence, the second term can be estimated 
as $\sum_{r=1}^{r_{\max}}n2^r(\log \log n)^2 = O(n\log n (\log \log n)^2)$.
If a data structure on the recursion level $r_{\max}$ contains $m$ 
elements, then it uses $O(m(\log \log n)^{1+\eps})$ words of space because 
$m\leq \log n$. All data structures on level $r_{\max}$ contain 
$O(n\log n)$ elements and use 
$O(n\log n (\log \log n)^{1+\eps})$ bits  of space. 
Thus the data structure uses $O(n(\log \log n)^{2})$ words of 
$\log n$ bits in total. 

The drawback of applying reduction to rank space on each recursion level is 
that we must pay a (higher than a constant) penalty for each point in the
 answer.
Consider a data structure $D_r$ on the $r$-th level of recursion, 
and let $P_r$ be the set of points stored in $D_r$.
Coordinates of any point stored in $D_r$ belong to the rank space 
of $P_r$. To obtain the point $p\in P$ that corresponds to a 
point $p_r\in P_r$ we need $O(r)=O(\log \log n)$ time.  
Hence, our data structure answers queries in 
$O((\log \log n)^3 + k\log \log n)$ time. 

The construction time can be estimated with the formula 
$$c(n)=O(n\log^3 n \log \log n) + 2(n^{1/2}/\log^p n)c(n^{1/2}\log^p n)$$ 
Therefore, $c(n)=O(n\log^3 n \log \log n)$.
\end{proof}

\begin{lemma}
There exists a data structure that answers (2,2,1)-sided queries 
on the $[n]\times [n]\times [n]$ grid in $O( (\log \log n)^3 + k\log \log n)$ time, 
uses $O(n (\log \log  n)^{3})$ space, 
and can be constructed in $O(n\log^{3} n \log \log n)$ time. 
\end{lemma}
\begin{proof} 
The data structure is the same as in Lemma~\ref{lemma:empt211} 
but in each $x$-slice there are two data structures for 
$(2,1,1)$-sided queries open in $+x$ and $-x$ directions, 
and in each $y$-slice there are two data structures for 
$(2,1,1)$-sided queries open in $+y$ and $-y$ direction.

The query is processed in the same way as in Lemma~\ref{lemma:two21}. 
The space usage can be analyzed in the same way as in
 Lemma~\ref{lemma:empt211}. 
Construction time can be estimated with the formula 
$c(n)=O(n\log^3 n \log \log n) +  2(n^{1/2}/\log^p n)c(n^{1/2}\log^p n)$
 and $c(n)=O(n\log^3 n \log \log n)$. 
\end{proof}

Finally, we can apply Lemma~\ref{lemma:transf} and reduction to rank space 
and obtain the 
data structure for three-dimensional orthogonal range 
reporting queries on the $[U]\times [U]\times [U]$ grid. 
\end{document}